\documentclass[journal]{IEEEtran}
\usepackage{amsmath,amssymb,amsthm}
\usepackage{graphicx}
\usepackage{booktabs}
\usepackage[T1]{fontenc}
\usepackage[hidelinks]{hyperref}
\hypersetup{pdftitle={The Single-Copy Quantum Bandit Is Classical: An Exact Spectral Collapse},pdfauthor={Siu Hin Ng}}

\newtheorem{theorem}{Theorem}
\newtheorem{proposition}{Proposition}
\newtheorem{lemma}{Lemma}
\newtheorem{corollary}{Corollary}
\theoremstyle{remark}
\newtheorem{remark}{Remark}
\newtheorem{openproblem}{Open Problem}

\newcommand{\Tr}{\operatorname{Tr}}
\newcommand{\kl}{\operatorname{kl}}
\newcommand{\KL}{\mathrm{KL}}
\newcommand{\D}{\mathcal D}
\newcommand{\Minf}{D_{M,\mathrm{inf}}}
\newcommand{\Dinf}{D_{\mathrm{inf}}}
\newcommand{\Kinf}{K_{\mathrm{inf}}}

\begin{document}

\title{The Single-Copy Quantum Bandit Is Classical:\\ An Exact Spectral Collapse}

\author{Siu Hin Ng\thanks{Siu Hin Ng is with Academia Sinica, Taipei, Taiwan. Email: ng\_siuhin@iis.sinica.edu.tw.}}

\maketitle

\begin{abstract}
We study multi-armed bandits whose arms supply unknown quantum states and whose mean rewards are defined by a known effect $F$. Each fresh copy may be measured by an arbitrary, adaptively chosen POVM. We prove an exact collapse: the measured relative entropy from an arm state to its confusing reward half-space equals the classical Burnetas--Katehakis functional of $F$'s spectral statistics, and the spectral measurement of $F$ together with an explicit least-favorable state, built from the derivative of the matrix logarithm, forms a saddle point of the underlying measurement game. Consequently, spectral measurement followed by classical KL-UCB is asymptotically instance-optimal among all consistent single-copy policies, in every finite dimension and for every reward effect under the stated nondegeneracy assumptions: adaptive and randomized measurement design cannot improve the leading logarithmic regret coefficient, and neither can storing copies for later single-copy measurement. A self-contained finite-time bound covers general effects and boundary distributions. Any further improvement must come from measurements outside the single-copy class. For consistent policies with arbitrary quantum memory, an amortized relative-entropy argument gives a converse with the Umegaki half-space divergence, and the two per-arm regret floors, $1/\Kinf$ and $1/\Dinf$, coincide if and only if the arm commutes with $F$. Whether the Umegaki floor is attainable involves a composite quantum Stein problem and a separate reduction to adaptive regret; we pose it as an open problem, with two-copy numerical evidence.
\end{abstract}

\section{Introduction}\label{sec:intro}

When a learner interacts with quantum systems, information is purchased by measurement, and a copy measured once is gone. This makes the multi-armed bandit --- the minimal model of the exploration--exploitation dilemma --- a natural probe of a basic question: \emph{in sequential learning of unknown quantum states, how much is measurement design worth?}

We consider $K$ arms, arm $k$ carrying an unknown state $\rho_k$ on a $d$-dimensional Hilbert space. Each pull hands the learner a single fresh copy of the pulled arm's state, which it may measure with any finite-outcome POVM, chosen adaptively (and possibly at random) from the entire observation history. Rewards are given by a known effect $F$: arm $k$ has mean reward $\mu_k=\Tr\rho_k F$, and performance is the pseudo-regret $R(T)=\sum_k \Delta_k\,\mathbb E[N_k(T)]$ with gaps $\Delta_k=\mu^\star-\mu_k$.

A priori one expects a rich measurement-design problem. The learner could measure the reward observable itself; or an informationally complete POVM; or adapt measurements to a running estimate; or randomize. Our main result says that, asymptotically, none of this matters: measuring the reward observable and running a classical index is optimal among all consistent single-copy strategies, in every dimension and for every effect.

This is not because measurement freedom is powerless. At the qubit instance of Section~\ref{sec:num}, a projective measurement tilted by $\pi/4$ from the eigenbasis of $F$ is fooled completely by the confusing set: the best arm's own state reproduces arm $k$'s outcome statistics, so the game value is zero (Section~\ref{sec:num}). A priori, then, the best single-copy measurement could have depended on the unknown state and required adaptivity. What the theorem shows is that one explicit state $\sigma^\star$, which need not commute with $\rho$, caps every POVM at once --- no single-copy measurement extracts more than $\Kinf$ against it --- while the reward measurement extracts at least $\Kinf$ against every state in the confusing set: the pair is a saddle point of the measurement game. The theorem was found numerically, when a nested optimization landed on the analytic witness to four decimals, and then proved; the proof is short because the Berta--Fawzi--Tomamichel variational objective is concave and $\sigma^\star$ is built to be its critical point. Its content is the exact identification, which fixes every single-copy regret constant at once. Measurements outside the single-copy class are the only resource that could lower the constant, and we determine exactly when the converse leaves them no room: the two floors coincide if and only if the arm commutes with $F$.

\emph{Contributions.}
(i)~\textbf{Spectral collapse and saddle point (Theorem~\ref{thm:collapse}, Corollary~\ref{cor:saddle}).} For any effect $F=\sum_i f_iP_i$, the measured relative entropy from $\rho$ to the half-space $\Lambda=\{\sigma:\Tr\sigma F\ge\mu^\star\}$ collapses exactly to the \emph{classical} constrained divergence $\Kinf$ of the spectral statistics $p_i=\Tr\rho P_i$ --- for a projector $F$, simply $\kl(\mu,\mu^\star)$. An explicit minimizer $\sigma^\star$ is obtained by applying the derivative of the matrix logarithm to $\rho$, and the pair (spectral measurement, $\sigma^\star$) is a saddle point of the game $(M,\sigma)\mapsto\KL(P^M_\rho\|P^M_\sigma)$: the supremum over POVMs and the infimum over the half-space interchange, as they do for any closed convex confusing set \cite{BBH21}, and the common value is $\Kinf$, attained by an explicit pair.
(ii)~\textbf{Single-copy optimality, with a full finite-time analysis (Theorems~\ref{thm:lbsc}--\ref{thm:finite}).} A change-of-measure converse over interaction transcripts shows every consistent single-copy policy pulls arm $k$ at least $(1-o(1))\ln T/\Kinf(k)$ times, \emph{including} adaptive and randomized measurement choice; the classical empirical-KL-UCB index \cite{MMS11,CGMMS13} run on spectral-measurement outcomes attains this constant, with explicit non-asymptotic constants (Section~\ref{sec:finite}) closing the achievability side, importing only classical concentration lemmas. Storing copies for later single-copy measurement does not lower the constant either (Remark~\ref{rem:delay}).
(iii)~\textbf{Quantum-memory converse (Theorem~\ref{thm:univ}).} An amortized argument on classical--quantum relative-entropy potentials gives the floor $\sum_k\Delta_k/\Dinf(k)$ for consistent quantum-memory strategies with arbitrary memory, capping the collective advantage by $\Dinf/\Kinf$.
(iv)~\textbf{Equality characterization (Corollary~\ref{cor:iff}).} The two floors coincide, $\Dinf(k)=\Kinf(k)$, if and only if $[\rho_k,F]=0$: the collective bracket is degenerate exactly in the commuting case, and every non-commuting arm leaves room, in principle, for a collective advantage.
(v)~\textbf{Two-copy evidence (Section~\ref{sec:n2}).} The block values obey $\Kinf\le V_n\le V_{2n}\le K^{(2n)}\le\Dinf$ (Remark~\ref{rem:blocks}). At two copies, product-spectral measurement attains $\Kinf$, a seeded local max--min search over two-copy projective bases ($U(4)$) returns no larger worst-case value, and an entangled candidate that appears advantageous on a two-parameter slice is defeated by a transverse alternative, which illustrates why slice searches can mislead. This is evidence for $V_2=\Kinf$ at the running instance, not a certificate over all POVMs or over other block sizes.

Attainability of the collective floor $\Dinf$ requires control of a composite quantum Stein problem and a separate reduction to adaptive regret. Regularized formulas for composite i.i.d. testing are available \cite{BBH21,LamiSanov25}; evaluating them for reward half-spaces and relating them to adaptive regret remain open (Open Problem~\ref{op:collective}).

The single-copy results are instantiated numerically on a qubit example (Section~\ref{sec:num}); the finite-time theorem is instantiated on a three-atom effect (Fig.~\ref{fig:ftsim}).

\emph{Related work.} This ``states-as-arms'' setting differs from the quantum bandits of Lumbreras, Haapasalo and Tomamichel \cite{LHT22}, where a \emph{single} unknown state is probed by \emph{known} observables acting as arms and the focus is minimax $\sqrt T$ behavior (see also \cite{LumPure24,LumAnyDim26,LumThesis25}); it is also distinct from quantum \emph{speedups} of classical bandits via quantum reward oracles \cite{Wan23}. Unknown quantum states already serve as actions in quantum contextual bandits \cite{BLT24}; here the objective is a fixed observable, and the result identifies the exact instance-dependent regret constant while optimizing over all single-copy measurements. Measured relative entropies govern optimal error exponents in sequential quantum hypothesis testing \cite{LTT22}, where measurement adaptivity can matter, and composite sequential quantum testing leads to minimal measured-relative-entropy exponents \cite{SPJ26}; our additional conclusion is the explicit collapse of the reward half-space problem to a fixed spectral measurement. The index of Section~\ref{sec:finite} is the finite-support empirical KL-UCB of \cite{MMS11,CGMMS13}, re-analyzed here with explicit constants on the spectral alphabet, and the collective converse follows the potential-function style of amortized channel discrimination \cite{BHKW20}.

\section{Model and preliminaries}\label{sec:model}

Let $\mathcal H$ be a Hilbert space, $\dim\mathcal H=d<\infty$, $\D(\mathcal H)$ the states. Fix a known effect $F$, $0\preceq F\preceq I$ (any bounded reward observable can be brought to this normalization affinely), with spectral resolution $F=\sum_{i=1}^{L} f_iP_i$, $f_i\in[0,1]$ distinct. An \emph{instance} is $\nu=(\rho_1,\dots,\rho_K)\in\D(\mathcal H)^K$ with means $\mu_k=\Tr\rho_kF$, a unique best arm $k^\star$, $\mu^\star=\mu_{k^\star}$, and $0<\mu_k$ for all $k$, $\mu^\star<\lambda_{\max}(F)$.

\emph{Policies.} A \emph{single-copy policy} ($\pi\in\Pi_1$): at round $t$ it selects an arm $A_t$ and a finite-outcome POVM $M_t$, both measurably (possibly at random) from the record $(A_s,M_s,X_s)_{s<t}$, measures the fresh copy of $\rho_{A_t}$ with $M_t$, and observes the outcome $X_t$. A \emph{quantum-memory policy} ($\pi\in\Pi_{\mathrm{all}}$) may instead store copies in a quantum register and, in each round, apply an arbitrary record-controlled quantum instrument to the entire register before selecting the next arm; arm selection is a (possibly randomized) function of the classical record. A policy is \emph{consistent} if $R(T)=o(T^a)$ for every $a>0$ on every instance. Regret is the pseudo-regret defined above; the optimal policy of Theorem~\ref{thm:match} measures $F$'s eigenbasis, so its outcomes $f_{X_t}$ are unbiased reward samples and the two regret readings coincide for it.

\emph{Divergences.} $D(\rho\|\sigma)=\Tr\rho(\ln\rho-\ln\sigma)$ is the Umegaki relative entropy \cite{Umegaki62}; $D_M(\rho\|\sigma)=\sup_{M}\KL(P^M_\rho\|P^M_\sigma)$ the measured relative entropy, supremum over finite-outcome POVMs; $\kl(x,y)$ the Bernoulli divergence. Data processing \cite{Lindblad75,Uhlmann77} gives $D_M\le D$, with equality iff $[\rho,\sigma]=0$ whenever $D(\rho\|\sigma)<\infty$ (Lemma~\ref{lem:finite-equality} below) and operational meaning through quantum Stein's lemma \cite{HiaiPetz91,OgawaNagaoka00}. We use the variational formula of \cite{BFT17}:
\begin{equation}\label{eq:bft}
D_M(\rho\|\sigma)\;=\;\sup_{\omega\succ 0}\;\Big\{\Tr\rho\ln\omega+1-\Tr\sigma\omega\Big\},
\end{equation}
whose objective is \emph{concave} in $\omega$ (operator concavity of $\ln$; linearity elsewhere), with gradient $\nabla_\omega=D\!\ln_\omega[\rho]-\sigma$, where $D\!\ln_\omega[X]=\int_0^\infty(\omega+s)^{-1}X(\omega+s)^{-1}ds$ is the (self-adjoint) derivative of the matrix logarithm.

\begin{lemma}[Equality at finite relative entropy]\label{lem:finite-equality}
If $D(\rho\|\sigma)<\infty$, then $D_M(\rho\|\sigma)=D(\rho\|\sigma)$ if and only if $[\rho,\sigma]=0$.
\end{lemma}
\begin{IEEEproof}
Restrict to $\operatorname{supp}\sigma$, on which $\sigma>0$. Theorem~2 of \cite{BFT17} identifies $D_M$ with the supremum over rank-one projective measurements; this supremum is attained, since the classical divergence is continuous on the compact set of orthonormal bases when $\sigma>0$, even if $\rho$ is singular. Let $\mathcal M$ be an attaining measurement channel. If $D_M=D$, equality in data processing gives a recovery channel $\mathcal R$ recovering both $\rho$ and $\sigma$ \cite{Petz86,JRSWW18}. For $0<\epsilon<1$, set $\rho_\epsilon=(1-\epsilon)\rho+\epsilon\sigma>0$. By linearity $\mathcal R\mathcal M$ also fixes $\rho_\epsilon$. Applying data processing to $\mathcal M$ and $\mathcal R$ shows that $D(\mathcal M(\rho_\epsilon)\|\mathcal M(\sigma))=D(\rho_\epsilon\|\sigma)$, hence $D_M(\rho_\epsilon\|\sigma)=D(\rho_\epsilon\|\sigma)$. Proposition~5 of \cite{BFT17}, now applied to two strictly positive states, implies $[\rho_\epsilon,\sigma]=0$, and thus $[\rho,\sigma]=0$. Conversely, a common eigenbasis attains $D_M=D$. Finiteness is essential: equality of two infinite divergences need not imply commutation.
\end{IEEEproof}

\emph{Confusing sets and exponents.} For suboptimal $k$ let $\Lambda_k=\{\sigma\in\D(\mathcal H):\Tr\sigma F\ge\mu^\star\}$, convex and compact, with nonempty interior by $\mu^\star<\lambda_{\max}(F)$. Write $p^{(k)}_i=\Tr\rho_kP_i$ and define the classical alphabet-restricted Burnetas--Katehakis functional \cite{BK96}
\begin{equation}\label{eq:kinf}
\Kinf(k)\;=\;\min\Big\{\KL\big(p^{(k)}\big\|q\big):\;q\in\Delta_L,\ \textstyle\sum_i f_iq_i\ge\mu^\star\Big\},
\end{equation}
together with the game value of a fixed POVM $M$ against the confusing set, $C^*_M(k):=\inf_{\sigma\in\Lambda_k}\KL(P^M_{\rho_k}\|P^M_\sigma)$, its optimum $C^*(k)=\sup_M C^*_M(k)$ (best single deterministic POVM), $\Minf(k)=\inf_{\sigma\in\Lambda_k}D_M(\rho_k\|\sigma)$ and $\Dinf(k)=\inf_{\sigma\in\Lambda_k}D(\rho_k\|\sigma)$.

\begin{proposition}[Exponent chain]\label{prop:chain}
$\kl(\mu_k,\mu^\star)\le\Kinf(k)\le C^*(k)\le\Minf(k)\le\Dinf(k)\le D(\rho_k\|\rho_{k^\star})$.
\end{proposition}

\begin{IEEEproof}
First, randomizing outcome $i$ to $\mathrm{Bern}(f_i)$ preserves means, so $\KL(p\|q)\ge\kl(\mu_k,\sum_if_iq_i)\ge\kl(\mu_k,\mu^\star)$ for feasible $q$. Second, as $\sigma$ ranges over $\Lambda_k$, the spectral statistics $q(\sigma)_i=\Tr\sigma P_i$ range over \emph{exactly} the feasible set of \eqref{eq:kinf}: feasibility is immediate, and any feasible $q$ is realized by $\sigma=\sum_iq_iP_i/\operatorname{rk}P_i$. Hence $\inf_{\sigma}\KL(P^{\{P_i\}}_{\rho_k}\|P^{\{P_i\}}_\sigma)=\Kinf(k)$, and $C^*$ is a supremum including this measurement. Third, $C^*\le\Minf$ by weak duality: $\inf_\sigma\KL(\cdot\|P^M_\sigma)\le\KL(\cdot\|P^M_{\sigma'})\le D_M(\rho_k\|\sigma')$ for every $\sigma'\in\Lambda_k$. The rest is $D_M\le D$ pointwise and $\rho_{k^\star}\in\Lambda_k$.
\end{IEEEproof}

A finite randomized choice of POVM, with the choice recorded, is itself a single POVM with labeled outcomes. The next theorem identifies a fixed spectral measurement attaining $C^*=\Minf$; a poorly chosen measurement can have a smaller game value $C^*_M$ without implying a gap between these optimized quantities.

\section{The spectral collapse}\label{sec:collapse}

\begin{theorem}[Spectral collapse]\label{thm:collapse}
Let $F=\sum_{i=1}^Lf_iP_i$ be an effect, $\rho\in\D(\mathcal H)$ with $p_i=\Tr\rho P_i$ and $\mu=\Tr\rho F$, and $\mu^\star\in(\mu,\lambda_{\max}(F))$. Then
\begin{equation}\label{eq:collapse}
\inf_{\sigma:\,\Tr\sigma F\ge\mu^\star} D_M(\rho\|\sigma)\;=\;\Kinf\big(p;f,\mu^\star\big),
\end{equation}
and the infimum is attained at the state
\begin{equation}\label{eq:witness}
\sigma^\star=\int_0^\infty\!(\omega+s)^{-1}\rho\,(\omega+s)^{-1}ds\;\oplus\!\!\sum_{i:\,p_i=0}\!q^\star_i\frac{P_i}{\operatorname{rk}P_i},
\end{equation}
where $q^\star$ optimizes \eqref{eq:kinf} and $\omega=\sum_{i:p_i>0}(p_i/q^\star_i)P_i$ on $V_+=\operatorname{ran}\sum_{i:p_i>0}P_i$. In particular, if $F$ is an orthogonal projector (of any rank, in any dimension),
\begin{equation}\label{eq:proj}
\inf_{\sigma:\,\Tr\sigma F\ge\mu^\star} D_M(\rho\|\sigma)\;=\;\kl(\mu,\mu^\star).
\end{equation}
Consequently $\Kinf(k)=C^*(k)=\Minf(k)$ for every effect $F$: the entire single-copy portion of the chain in Proposition~\ref{prop:chain} is one number.
\end{theorem}

\begin{IEEEproof}
($\ge$) was shown in Proposition~\ref{prop:chain}: every $\sigma\in\Lambda$ induces a feasible $q(\sigma)$, and the spectral POVM is one competitor in the supremum defining $D_M$.

($\le$). An optimizer $q^\star$ of \eqref{eq:kinf} exists (compact feasible set, lower-semicontinuous objective) and is finite: $q=(1-\varepsilon)p+\varepsilon\,\delta_{\mathrm{top}}$ is feasible for $\varepsilon\ge(\mu^\star-\mu)/(f_{\max}-\mu)<1$ and has $\KL(p\|q)<\infty$. Finiteness forces $q^\star_i>0$ wherever $p_i>0$, and optimality forces the constraint to bind, $\sum_if_iq^\star_i=\mu^\star$. Since $\Tr\rho P_i=0$ implies $P_i\rho=0$, the state $\rho$ is supported on $V_+$, on which $\omega\succ0$. Set $\tilde\sigma=D\!\ln_\omega[\rho]$ on $V_+$. Then: (i) $\tilde\sigma\succeq0$, as an integral of congruences of $\rho\succeq0$ with integrand $O(s^{-2})$, and $\ker\tilde\sigma\subseteq\ker\rho$ (if $\tilde\sigma u=0$ then $\rho^{1/2}(\omega+s)^{-1}u=0$ for every $s\ge0$; linear independence of $s\mapsto(\omega_a+s)^{-1}$ across distinct eigenvalues of $\omega$ then forces $\rho^{1/2}u=0$), so $\operatorname{supp}\rho\subseteq\operatorname{supp}\sigma^\star$ and $D_M(\rho\|\sigma^\star)<\infty$ a priori; (ii) using self-adjointness of $D\!\ln_\omega$ and $[P_i,\omega]=0$, $\Tr\tilde\sigma P_i=\Tr[\rho\,P_i/\omega_i]=q^\star_i$ for $p_i>0$, hence $\Tr\tilde\sigma=1-m$ with $m=\sum_{i:p_i=0}q^\star_i$. Thus $\sigma^\star$ in \eqref{eq:witness} is a state with $\Tr\sigma^\star F=\sum_if_iq^\star_i=\mu^\star$, so $\sigma^\star\in\Lambda$.

We evaluate $D_M(\rho\|\sigma^\star)$ with \eqref{eq:bft}. Write $g(\omega')=\Tr\rho\ln\omega'+1-\Tr\sigma^\star\omega'$ and $\Pi_\pm$ for the projectors onto $V_+$ and $V_+^\perp$. For any $\omega'\succ0$, the operator Jensen inequality for the unital positive compression $X\mapsto\Pi_+X\Pi_+$ and the operator concave $\ln$ gives $\Pi_+(\ln\omega')\Pi_+\preceq\ln(\Pi_+\omega'\Pi_+)$ on $V_+$; since $\rho=\Pi_+\rho\Pi_+$ and $\sigma^\star$ is block diagonal,
\[
g(\omega')\;\le\;\big[\Tr\rho\ln\omega_+ +1-\Tr\tilde\sigma\,\omega_+\big]\;-\;\Tr\big[(\sigma^\star-\tilde\sigma)\,\Pi_-\omega'\Pi_-\big]
\]
with $\omega_+=\Pi_+\omega'\Pi_+\succ0$; the last term is $\le0$ and vanishes as the $V_+^\perp$ block tends to $0^+$. So the supremum in \eqref{eq:bft} equals its restriction to $V_+$, where the objective is concave with gradient $D\!\ln_{\omega_+}[\rho]-\tilde\sigma$, which \emph{vanishes at $\omega_+=\omega$ by construction of $\tilde\sigma$}. A critical point of a concave function is its global maximum, so
\[
D_M(\rho\|\sigma^\star)=g(\omega)=\underbrace{\sum_{i:p_i>0}p_i\ln\frac{p_i}{q^\star_i}}_{=\,\Kinf}+1-\underbrace{\Tr\big[\rho\,D\!\ln_\omega[\omega]\big]}_{=\,\Tr[\rho\,\omega\omega^{-1}]=1},
\]
which is $\Kinf$. For projector $F$ the classical program \eqref{eq:kinf} is the two-atom case; the tilt $q^\star=(\mu^\star,1-\mu^\star)$ on the atoms $(F,I-F)$ gives $\Kinf=\kl(\mu,\mu^\star)$, matched by the Bernoulli data-processing lower bound.
\end{IEEEproof}

\begin{corollary}[Saddle point of the measurement game]\label{cor:saddle}
Let $M^\star=\{P_i\}$ be the spectral measurement of $F$, $\sigma^\star$ the state \eqref{eq:witness}, and $\Lambda=\{\sigma\in\D(\mathcal H):\Tr\sigma F\ge\mu^\star\}$. For every finite-outcome POVM $M$ and every $\sigma\in\Lambda$,
\[
\KL\big(P^M_\rho\big\|P^M_{\sigma^\star}\big)\ \le\ \Kinf(p;f,\mu^\star)\ \le\ \KL\big(P^{M^\star}_\rho\big\|P^{M^\star}_\sigma\big).
\]
Hence $(M^\star,\sigma^\star)$ is a saddle point of $g(M,\sigma):=\KL(P^M_\rho\|P^M_\sigma)$ on $\{\text{POVMs}\}\times\Lambda$, and
\[
\sup_M\inf_{\sigma\in\Lambda}g(M,\sigma)\;=\;\inf_{\sigma\in\Lambda}\sup_M g(M,\sigma)\;=\;\Kinf(p;f,\mu^\star).
\]
\end{corollary}

\begin{IEEEproof}
The left inequality is $\KL(P^M_\rho\|P^M_{\sigma^\star})\le D_M(\rho\|\sigma^\star)=\Kinf$, by Theorem~\ref{thm:collapse}. For the right one, $P^{M^\star}_\sigma=q(\sigma)$ is feasible for \eqref{eq:kinf} whenever $\sigma\in\Lambda$ (Proposition~\ref{prop:chain}), so $\KL(p\|q(\sigma))\ge\Kinf$; at $(M,\sigma)=(M^\star,\sigma^\star)$ both inequalities hold with equality, since $q(\sigma^\star)=q^\star$. Finally
\begin{align*}
\sup_M\inf_{\sigma}g\ &\ge\ \inf_{\sigma}g(M^\star,\sigma)\ \ge\ \Kinf\ =\ D_M(\rho\|\sigma^\star)\\
&=\ \sup_M g(M,\sigma^\star)\ \ge\ \inf_{\sigma}\sup_M g,
\end{align*}
while $\sup\inf\le\inf\sup$ always. The interchange itself is a special case of the minimax lemma of \cite{BBH21} (Lemma~A.2 there, for a closed convex alternative set); the content of the corollary is the value $\Kinf$ and the explicit pair attaining it.
\end{IEEEproof}

\begin{remark}[Classical consistency]\label{rem:classical}
If $[\rho,F]=0$, then $\sigma^\star$ is the classical reverse-KL projection and \eqref{eq:collapse} restates known facts about $\Kinf$ \cite{BK96,HondaTakemura15}; for indicator rewards, $\Kinf$ reduces to $\kl$ of the means by Bernoulli sufficiency. Theorem~\ref{thm:collapse} is the exact quantum generalization --- and it says the generalization adds \emph{nothing} on the single-copy side.
\end{remark}

\begin{remark}[Structure of the witness]\label{rem:witness}
The witness rescales each coherence $\rho_{ab}$ (in the eigenbasis of $\omega$) by the reciprocal logarithmic mean $1/L(\omega_a,\omega_b)$ --- the Kubo--Mori geometry. At $\sigma^\star$, the optimal measurement is the spectral measurement of $F$ itself: the supremum in \eqref{eq:bft} is approached along $F$-diagonal $\omega$, and attained at $\omega$ when $q^\star$ has no off-support mass. Everything one copy can say against the confusing set is said by measuring the reward observable.
\end{remark}

\section{Matching bounds for single-copy policies}\label{sec:single}

\begin{theorem}[Single-copy converse; adaptivity is useless]\label{thm:lbsc}
For every consistent $\pi\in\Pi_1$ and every suboptimal arm $k$,
\[
\liminf_{T\to\infty}\ \frac{\mathbb E[N_k(T)]}{\ln T}\;\ge\;\frac{1}{\Minf(k)}\;\overset{\mathrm{Thm.}~\ref{thm:collapse}}{=}\;\frac{1}{\Kinf(k)}.
\]
\end{theorem}

\begin{IEEEproof}
Fix $\sigma\in\operatorname{int}\Lambda_k$ and the confused instance $\nu'$ in which arm $k$ carries $\sigma$. The transcript law factorizes round by round; by the chain rule and the definition of $D_M$ as a supremum over POVMs,
\[
\KL\big(\mathbb P^T_\nu\big\|\mathbb P^T_{\nu'}\big)=\sum_{t\le T}\mathbb E_\nu\Big[\mathbf 1\{A_t{=}k\}\,\KL\big(P^{M_t}_{\rho_k}\big\|P^{M_t}_\sigma\big)\Big]
\]
\[
\le\;\mathbb E_\nu[N_k(T)]\cdot D_M(\rho_k\|\sigma),
\]
\emph{regardless of how the $M_t$ were chosen}. The fundamental inequality of \cite{GMS19} with $Z=N_k(T)/T$ and consistency on both instances give $\KL(\mathbb P_\nu\|\mathbb P_{\nu'})\ge(1-o(1))\ln T$. It remains to drive $D_M(\rho_k\|\sigma)$ to $\Minf(k)$ within $\operatorname{int}\Lambda_k$: take $\sigma_\varepsilon=(1-\varepsilon)\sigma^\star+\varepsilon\tau$ with $\tau\succ0$, $\Tr\tau F>\mu^\star$ (which exists by $\mu^\star<\lambda_{\max}(F)$); then $\Tr\sigma_\varepsilon F>\mu^\star$ and, by convexity of $D_M$ in its second argument, $D_M(\rho_k\|\sigma_\varepsilon)\le(1-\varepsilon)\Kinf(k)+\varepsilon\,D_M(\rho_k\|\tau)$ with the last term finite. Let $\varepsilon\downarrow0$.
\end{IEEEproof}

\begin{remark}[Storing copies does not help without joint measurements]\label{rem:delay}
Let $\Pi_{\mathrm{delay}}\supseteq\Pi_1$ consist of policies that may store copies in a register and measure each stored copy at that or any later round, once, with a finite-outcome POVM chosen (possibly at random) from the classical record at that time, but never jointly with another copy. The proof of Theorem~\ref{thm:lbsc} applies verbatim: the transcript law still factorizes over measurement events; each event measures one copy of a known arm, so its contribution to the transcript divergence is at most $D_M(\rho_k\|\sigma)$ when that arm is $k$ and zero otherwise, unmeasured copies contribute nothing, and the number of measured copies of arm $k$ by round $T$ is at most $N_k(T)$. Hence the floor $1/\Kinf(k)$ of Theorem~\ref{thm:lbsc} holds for every consistent policy in $\Pi_{\mathrm{delay}}$: any improvement of the leading constant requires measurements outside this class, acting on several copies at once or returning to a copy already measured.
\end{remark}

\begin{theorem}[Matching achievability and optimality]\label{thm:match}
Let \textsc{Spectral-KL-UCB} measure every pulled copy with the spectral POVM $\{P_i\}$ of $F$ and play the index policy of Section~\ref{sec:finite}. Then for every suboptimal $k$,
$\limsup_T\mathbb E[N_k(T)]/\ln T\le1/\Kinf(k)$; combined with Theorem~\ref{thm:lbsc},
\[
\lim_{T\to\infty}\frac{R(T)}{\ln T}\;=\;\sum_{k\ne k^\star}\frac{\Delta_k}{\Kinf(k)}
\qquad\big(\overset{F\ \mathrm{proj.}}{=}\ \sum_{k\ne k^\star}\tfrac{\Delta_k}{\kl(\mu_k,\mu^\star)}\big),
\]
and \textsc{Spectral-KL-UCB} is asymptotically instance-optimal in $\Pi_1$. \emph{No consistent single-copy strategy --- adaptive, randomized, tomographic, or otherwise --- improves the leading constant over measuring the reward observable itself.}
\end{theorem}

The induced problem \emph{is} a classical bandit whose arm distributions live on the known finite alphabet $\operatorname{spec}F$, and by Proposition~\ref{prop:chain} the quantum confusing set induces exactly the alphabet-restricted classical confusing set, so the constants agree on both sides. Section~\ref{sec:finite} proves Theorem~\ref{thm:match} with explicit finite-time constants.

\begin{remark}[The known alphabet helps]
Nature is confined to spectral distributions on $\operatorname{spec}F$, so the relevant $\Kinf$ is the alphabet-restricted one --- in general \emph{larger} (hence a smaller regret constant) than the unrestricted $[0,1]$-support functional of \cite{HondaTakemura15} (cf.\ the finite-support analyses of \cite{MMS11,CGMMS13}). The quantum structure thus mildly \emph{helps} the learner relative to generic bounded classical rewards.
\end{remark}

\section{Finite-time analysis of {\sc Spectral-KL-UCB}}\label{sec:finite}

This section proves Theorem~\ref{thm:match} with explicit constants for a general effect $F$, closing the achievability side of Section~\ref{sec:single}; the only imports are the classical concentration lemmas below.

\emph{The index.} Let $n_k(t)=N_k(t-1)$ be the number of pulls before action $A_t$, and form $\hat q_k(t)$ from those samples. After pulling every arm once, at rounds $t\ge\max\{K+1,3\}$ play $\arg\max_k U_k(t)$ with
\begin{align*}
U_k(t)=\max\big\{\langle f,q\rangle:\ &q\in\Delta_L,\\
&n_k(t)\,\KL(\hat q_k(t)\|q)\le\ell(t,n_k(t))\big\},
\end{align*}
\[
\ell(t,n)\;=\;\ln t+3\ln\ln t+\tfrac{L-1}{2}\ln n+c_L,
\]
where $\hat q_k(t)$ is the empirical spectral distribution of arm $k$, $q_k:=p^{(k)}$ is its population law (so that $\Kinf(k)=\Kinf(q_k)$ in the notation of Lemma~\ref{lem:dual}), and $c_L$ is the Krichevsky--Trofimov constant of Lemma~\ref{lem:mixture} below. This is the empirical (finite-support) KL-UCB index of \cite{MMS11,CGMMS13} with a time-uniform mixture threshold; Theorem~\ref{thm:finite} re-derives its guarantee with explicit constants on the spectral alphabet.

\subsection{Three lemmas}

\begin{lemma}[Dual form and subgradient]\label{lem:dual}
Let $p\in\Delta_L$, $\mu:=\langle f,p\rangle<\mu^\star<f_{\max}$, $\bar\lambda:=1/(f_{\max}-\mu^\star)$. Then
\[
\Kinf(p)\;=\;\max_{\lambda\in[0,\bar\lambda]}\ \sum_{i=1}^L p_i\ln\big(1-\lambda(f_i-\mu^\star)\big),
\]
with $p_i\ln(\cdot)=0$ when $p_i=0$. The function $\Kinf$ is convex and lower-semicontinuous in $p$. Let $S=\operatorname{supp}p$ and choose a maximizer $\lambda^\star$. On the face $\Delta_S=\{p':\operatorname{supp}p'\subseteq S\}$, set
\[
g_i=\ln(1-\lambda^\star(f_i-\mu^\star)),\quad i\in S,
\qquad G(p)=\max_{i\in S}|g_i|.
\]
These coefficients are finite, including when $\lambda^\star=\bar\lambda$, and give the supporting inequality $\Kinf(p')\ge\Kinf(p)+\sum_{i\in S}g_i(p'_i-p_i)$ for $p'\in\Delta_S$. If $p$ has positive mass on a top-value atom, then $\lambda^\star<\bar\lambda$ and the subgradient is finite on the full simplex.
\end{lemma}

\begin{IEEEproof}
Slater's condition holds since $\mu^\star<f_{\max}$. Use the Lagrangian
\[
\mathcal L(q,\lambda,\alpha)=\KL(p\|q)
+\alpha(\textstyle\sum_iq_i-1)
+\lambda(\mu^\star-\textstyle\sum_if_iq_i),
\]
where $\lambda\ge0$ and $q_i\ge0$. The KKT conditions give $(\alpha-\lambda f_i)q_i=p_i$ on $S$; off $S$, $\alpha-\lambda f_i\ge0$ and $(\alpha-\lambda f_i)q_i=0$. Since $\mu<\mu^\star$, the mean constraint binds and $\lambda>0$. Summing these identities gives $\alpha-\lambda\mu^\star=1$, hence
\[
q_i^\star=\frac{p_i}{1-\lambda^\star(f_i-\mu^\star)},\quad i\in S.
\]
Nonnegativity of the top-atom coefficient gives $\lambda^\star\le\bar\lambda$. Off-support mass is possible only at a top atom when equality holds. Substitution in the dual, using $\alpha=1+\lambda\mu^\star$, gives the stated formula. Convexity and lower semicontinuity follow by writing the value as the supremum of the finite affine objectives with $0\le\lambda<\bar\lambda$, including their endpoint limits. If a top atom has positive p-mass, the endpoint objective tends to $-\infty$, so a maximizer is interior. Otherwise, even at the endpoint all coefficients indexed by $S$ are finite. Evaluating the dual for $p'\in\Delta_S$ at the maximizer for $p$ gives the supporting inequality. This uses an affine supporting function, not an assertion that $\Kinf$ is affine on the face.
\end{IEEEproof}

The only property used downstream is the one-sided inequality $\Kinf(\hat q)\ge\Kinf(q_k)-G_k\|\hat q-q_k\|_1$ for empirical $\hat q$ of samples from $q_k$, with $G_k:=G(q_k)$.

\begin{lemma}[Time-uniform multinomial deviation]\label{lem:mixture}
Let $X_1,X_2,\dots$ be i.i.d.\ from $q\in\Delta_L$ and $\hat q_n$ the empirical distribution. For every $\delta\in(0,1)$,
\[
\mathbb P\Big(\exists n\ge1:\ n\,\KL(\hat q_n\|q)\ \ge\ \ln\tfrac1\delta+\tfrac{L-1}{2}\ln n+c_L\Big)\ \le\ \delta,
\]
where $c_L:=\sup_{n\ge1}\big(R_L(n)-\tfrac{L-1}{2}\ln n\big)<\infty$ is the (finite, classical \cite{KrichevskyTrofimov81,XieBarron00}) worst-case log-loss regret excess of the Krichevsky--Trofimov (Dirichlet-$\tfrac12$) mixture code over $n$ symbols on alphabet size $L$.
\end{lemma}

\begin{IEEEproof}
Let $W_n=\dfrac{\int_{\Delta_L}\prod_{j\le n}w_{X_j}\,d\pi_{1/2}(w)}{\prod_{j\le n}q_{X_j}}$ with $\pi_{1/2}=\mathrm{Dir}(\tfrac12,\dots,\tfrac12)$. Set $W_0=1$. Under $q$, $(W_n)$ is a nonnegative supermartingale: its conditional expected multiplier is the mixture predictive probability assigned to $\operatorname{supp}q$, which is at most one. It is a martingale when $q$ has full support. Thus Ville's inequality \cite{Ville39} gives $\mathbb P(\exists n:W_n\ge1/\delta)\le\delta$. Pointwise, $\ln W_n\ge n\,\KL(\hat q_n\|q)-R_L(n)$, since the maximum-likelihood weight vector is $\hat q_n$. The event in the statement thus implies $\ln W_n\ge\ln(1/\delta)$, and the claim follows with $R_L(n)\le\tfrac{L-1}{2}\ln n+c_L$.
\end{IEEEproof}

\begin{remark}[Exact constants for two and three atoms]
The worst-case KT redundancy is attained when all $n$ observations have the same symbol, giving
\[
R_L(n)=\ln\frac{\Gamma(n+L/2)\Gamma(1/2)}{\Gamma(L/2)\Gamma(n+1/2)}.
\]
Indeed, the count-dependent part is $\sum_i h(n_i)$, where $h(m)=m\ln m-\ln\Gamma(m+1/2)+\ln\Gamma(1/2)$ and $h(0)=0$. Its discrete increments increase: the map $m\mapsto h(m+1)-h(m)=(m+1)\ln(m+1)-m\ln m-\ln(m+1/2)$ is continuous on $[0,\infty)$ with value $\ln2$ at $m=0$ and has derivative $\ln((m+1)/m)-1/(m+1/2)>0$ on $(0,\infty)$, so it is increasing on $[0,\infty)$. Transferring counts from a smaller nonzero cell to a larger one therefore increases the sum, so a vertex maximizes it. For $L=3$, $R_3(n)=\ln(2n+1)$, whence $\sup_{n\ge1}(R_3(n)-\ln n)=\ln3$. For $L=2$, $A_n:=\exp(R_2(n))/\sqrt n$ satisfies $A_{n+1}/A_n=\sqrt{n(n+1)}/(n+1/2)<1$ and $A_1=2$, proving $c_2=\ln2$. These are uniform bounds, not conclusions from a finite numerical scan.
\end{remark}

\begin{lemma}[$L_1$ concentration \cite{WOSVW03}]\label{lem:l1}
For i.i.d.\ samples from $q\in\Delta_L$ and any $\varepsilon>0$: $\ \mathbb P(\|\hat q_n-q\|_1\ge\varepsilon)\le(2^L-2)\,e^{-n\varepsilon^2/2}$.
\end{lemma}

\subsection{The finite-time theorem}

\begin{theorem}[Finite-time regret bound]\label{thm:finite}
Under the model of Section~\ref{sec:model}, for every integer $T\ge3$, every suboptimal arm $k$ and every $\delta\in(0,\Kinf(q_k))$, with $G_k$ from Lemma~\ref{lem:dual} and $\varepsilon_k:=\delta/G_k$:
\[
\mathbb E[N_k(T)]\ \le\ n_\delta(T)\;+\;\frac{2^L-2}{e^{\varepsilon_k^2/2}-1}\;+\;3.
\]
With $a=(L-1)/2$ and $b_k=\Kinf(q_k)-\delta$, $n_\delta(T)$ is the smallest integer $n\ge\max\{1,\lceil a/b_k\rceil\}$ such that $nb_k\ge\ell(T,n)$. On this branch $nb_k-a\ln n$ is nondecreasing, so the inequality holds at every subsequent integer.
\end{theorem}

\begin{IEEEproof}
Initialization and any action before $t=3$ account for at most $2$ pulls of a fixed arm. At subsequent rounds, since $A_t=k$ implies $U_k(t)\ge U_{k^\star}(t)$,
\[
\mathbf1\{A_t=k\}\ \le\ \mathbf1\{U_{k^\star}(t)<\mu^\star\}\ +\ \mathbf1\{A_t=k,\ U_k(t)\ge\mu^\star\}.
\]
\emph{Term (A).} $U_{k^\star}(t)<\mu^\star$ means $n_{k^\star}(t)\,\KL(\hat q_{k^\star}\|q_{k^\star})>\ell(t,n_{k^\star}(t))$; by Lemma~\ref{lem:mixture} with $\delta_t=1/(t\ln^3t)$ this has probability $\le\delta_t$, and $\sum_{t\ge3}\delta_t\le\frac{1}{3\ln^33}+\int_3^\infty\frac{dt}{t\ln^3t}<1$.

\emph{Term (B).} Rounds with $A_t=k$ and $n_k(t)=n$ occur at most once per value of $n$ (the count increments on the pull), so $\sum_{t}\mathbf1\{A_t=k,\ U_k(t)\ge\mu^\star\}\le 1+\sum_{n\ge1}\mathbf1\{\Kinf(\hat q_{k,n})\le\ell(T,n)/n\}$, with $\hat q_{k,n}$ the empirical distribution of the first $n$ spectral outcomes of arm $k$. If $U_k(t)\ge\mu^\star$ then $\Kinf(\hat q_k(t))\le\ell(t,n_k(t))/n_k(t)\le\ell(T,n_k(t))/n_k(t)$. By the definition of the crossing index, $\ell(T,n)/n\le\Kinf(q_k)-\delta$ for every $n\ge n_\delta(T)$, so the event forces $\Kinf(\hat q_{k,n})\le\Kinf(q_k)-\delta$, hence (Lemma~\ref{lem:dual}'s one-sided bound) $\|\hat q_{k,n}-q_k\|_1\ge\varepsilon_k$. The indicators with $n<n_\delta(T)$ contribute at most $n_\delta(T)-1$, and summing the tail bound of Lemma~\ref{lem:l1} over $n\ge n_\delta(T)$ gives $\mbox{$(2^L-2)/(e^{\varepsilon_k^2/2}-1)$}$. Adding the two initial pulls and Term~(A) yields the claim.
\end{IEEEproof}

\begin{corollary}[Closing Theorem~\ref{thm:match}]\label{cor:closeasymp}
For fixed $\delta$, the crossing index is $n_\delta(T)=(\ln T+O(\ln\ln T))/(\Kinf(q_k)-\delta)$. Taking $\delta=\delta_T\downarrow0$ slowly (e.g.\ $\delta_T=1/\ln\ln T$) in Theorem~\ref{thm:finite} gives $\limsup_T\mathbb E[N_k(T)]/\ln T\le1/\Kinf(k)$ for every suboptimal $k$, closing Theorem~\ref{thm:match} using only Lemmas~\ref{lem:dual}--\ref{lem:l1}, of which Lemmas~\ref{lem:mixture}--\ref{lem:l1} are the classical imports.
\end{corollary}

For projector rewards ($L=2$), the index optimization is Bernoulli KL-UCB; sharper analyses for the standard KL-UCB threshold are available in \cite{GarivierCappe11,CGMMS13}. Theorem~\ref{thm:finite}'s uniform statement covers general $F$. Section~\ref{sec:num} instantiates the constants and compares the theorem's envelope with simulation.

\section{Quantum memory: a converse, a cap, and an open problem}\label{sec:collective}

\begin{theorem}[Universal converse for quantum-memory policies]\label{thm:univ}
For every consistent $\pi\in\Pi_{\mathrm{all}}$ and suboptimal $k$,
\[
\liminf_{T\to\infty}\ \frac{\mathbb E[N_k(T)]}{\ln T}\;\ge\;\frac{1}{\Dinf(k)},
\]
and consequently $\liminf_T R(T)/\ln T\ge\sum_{k\ne k^\star}\Delta_k/\Dinf(k)$.
\end{theorem}

\begin{IEEEproof}
Fix $\sigma\in\operatorname{int}\Lambda_k$ and $\nu'$ as before. Model the protocol as acting on a classical--quantum register: classical record $X$ (all past outcomes and coins) and quantum memory $Q$ (all stored, unmeasured copies). Each round consists of (a) a record-controlled instrument on $Q$ appending its outcome to $X$, (b) an arm choice $A_t=a_t(X)$, and (c) appending a fresh copy of the pulled arm's state to $Q$; a final instrument produces the terminal record. Any physical strategy has this form: arm selection is a classical act, so any dependence on quantum data is mediated by a preceding instrument. Let $\bar\rho_t,\bar\sigma_t$ be the joint $cq$ states under $\nu,\nu'$ and $\Phi_t=D(\bar\rho_t\|\bar\sigma_t)$. Step (a) and the final step are channels applied identically under both models (the instrument may act jointly on stored copies of all arms; data processing needs only that the same channel is applied under both models), so $\Phi$ does not increase. For step (c), writing the $cq$ states as $\sum_xp(x)|x\rangle\langle x|\otimes\rho_x$ and $\sum_xq(x)|x\rangle\langle x|\otimes\sigma_x$ and using the classical--quantum decomposition $D(\bar\rho\|\bar\sigma)=\KL(p\|q)+\sum_xp(x)D(\rho_x\|\sigma_x)$ together with additivity of $D$ on product states,
\begin{align*}
\Phi_{t^+}-\Phi_t&=\textstyle\sum_xp(x)\,D\big(\rho_{a_t(x)}\big\|\sigma'_{a_t(x)}\big)\\
&=\Pr\nolimits_\nu[A_t{=}k]\;D(\rho_k\|\sigma),
\end{align*}
since $\sigma'_j=\rho_j$ for $j\ne k$. Summing from $\Phi_0=0$ (identical initial registers), $\Phi_T\le\mathbb E_\nu[N_k(T)]\,D(\rho_k\|\sigma)$. The terminal record is a classical function of the final $cq$ state, so $\KL(\mathbb P_\nu\|\mathbb P_{\nu'})\le\Phi_T$, and \cite{GMS19} plus consistency give $\ge(1-o(1))\ln T$ as before. Conclude with the same interior perturbation, using convexity of $D(\rho_k\|\cdot)$ and finiteness at a full-rank $\tau$.
\end{IEEEproof}

The potential-function style mirrors amortized converses in quantum channel discrimination \cite{BHKW20}; for fresh i.i.d.\ copies the append step is an exact additivity, so no amortization gap arises.

\begin{corollary}[Cap on the collective advantage]\label{cor:cap}
For each suboptimal $k$, the ratio between the forced exploration of an optimal single-copy policy and that permitted to \emph{any consistent} quantum-memory policy is at most $\Dinf(k)/\Kinf(k)$. If $[\rho_k,F]=0$ the classical tilt lies in $\Lambda_k$ and witnesses $\Dinf(k)=\Kinf(k)$: no collective advantage (Corollary~\ref{cor:iff} shows that the condition is also necessary). At the non-commuting instance of Section~\ref{sec:num} the cap is $1.346$.
\end{corollary}

\subsection{When does the cap equal one?}\label{sec:remark3}

Commutation of the arm with the reward observable is sufficient for the cap to equal one (Corollary~\ref{cor:cap}); we now show that it is also necessary. The first step characterizes equality through an optimal confusing state.

\begin{proposition}[Commuting optimal alternative]\label{prop:remark3}
$\Dinf(k)=\Kinf(k)$ if and only if there exists $\hat\sigma\in\Lambda_k$ that simultaneously (i) minimizes $D_M(\rho_k\|\cdot)$ over $\Lambda_k$ (i.e.\ achieves the value $\Kinf(k)$) and (ii) commutes with $\rho_k$. When this holds, $\hat\sigma$ also minimizes $D(\rho_k\|\cdot)$ over $\Lambda_k$.
\end{proposition}

\begin{IEEEproof}
($\Leftarrow$) If such $\hat\sigma$ exists: $D(\rho_k\|\hat\sigma)=D_M(\rho_k\|\hat\sigma)$ by the commuting equality condition, and $D_M(\rho_k\|\hat\sigma)=\Kinf(k)$ by hypothesis (i). Then $\Dinf(k)\le D(\rho_k\|\hat\sigma)=\Kinf(k)\le\Dinf(k)$ (the last inequality is Proposition~\ref{prop:chain}), forcing equality throughout.

($\Rightarrow$) Suppose $\Dinf(k)=\Kinf(k)$. Let $\hat\sigma$ be any minimizer of $D(\rho_k\|\cdot)$ over $\Lambda_k$ (exists by compactness and lower semicontinuity). Then
\[
\Kinf(k)\ \le\ D_M(\rho_k\|\hat\sigma)\ \le\ D(\rho_k\|\hat\sigma)=\Dinf(k)=\Kinf(k),
\]
the first inequality since $\hat\sigma\in\Lambda_k$ and $\Kinf(k)=\Minf(k)$ is an infimum over $\Lambda_k$, the second by data processing. The sandwich forces $D_M(\rho_k\|\hat\sigma)=\Kinf(k)$ [proving (i)] and $D_M(\rho_k\|\hat\sigma)=D(\rho_k\|\hat\sigma)$, which by the equality condition gives $[\rho_k,\hat\sigma]=0$ [proving (ii)].
\end{IEEEproof}

Proposition~\ref{prop:remark3} concerns the \emph{optimal confusing state}, not $F$: the constraint $\Tr\hat\sigma F\ge\mu^\star$ only sees $\hat\sigma$'s diagonal entries against $F$, so a priori $\hat\sigma$ might commute with $\rho_k$ by coincidence while $F$ does not. The next corollary rules this out.

\begin{corollary}[The cap equals one iff the arm commutes with the reward]\label{cor:iff}
$\Dinf(k)=\Kinf(k)$ if and only if $[\rho_k,F]=0$.
\end{corollary}

\begin{IEEEproof}
($\Leftarrow$) is Corollary~\ref{cor:cap}. ($\Rightarrow$) Write $\rho=\rho_k$, $p_i=\Tr\rho P_i$ and $S=\{i:p_i>0\}$; let $q^\star$ optimize \eqref{eq:kinf}. By Lemma~\ref{lem:dual} the constraint binds with multiplier $\lambda^\star>0$ and $p_i/q^\star_i=1-\lambda^\star(f_i-\mu^\star)$ for $i\in S$, so these ratios are distinct across $i\in S$ because the $f_i$ are distinct. By Proposition~\ref{prop:remark3} there is $\hat\sigma\in\Lambda_k$ with $[\rho,\hat\sigma]=0$ and $D(\rho\|\hat\sigma)=D_M(\rho\|\hat\sigma)=\Kinf(k)<\infty$. Take a common eigenbasis $\{e_j\}$, write $r_j=\langle e_j|\rho|e_j\rangle$, $s_j=\langle e_j|\hat\sigma|e_j\rangle$, and let $W(i|j)=\langle e_j|P_i|e_j\rangle$, a stochastic matrix; then $p=rW$, the spectral statistics of $\hat\sigma$ are $q:=sW$, and $D(\rho\|\hat\sigma)=\KL(r\|s)$, so $s_j>0$ whenever $r_j>0$.

(a) $q$ is feasible for \eqref{eq:kinf} and $\Kinf(k)\le\KL(p\|q)\le D_M(\rho\|\hat\sigma)=\Kinf(k)$, so $q$ is optimal. Since $\KL(p\|\cdot)$ is strictly convex in the coordinates indexed by $S$ and the feasible set is convex, $q_i=q^\star_i$ for all $i\in S$.

(b) $\KL(r\|s)=\KL(p\|q)$ is equality in classical data processing through $W$. Writing
\[
\KL(r\|s)=\sum_i\sum_j r_jW(i|j)\ln\frac{r_jW(i|j)}{s_jW(i|j)}
\]
and applying the log-sum inequality to each inner sum gives $\KL(r\|s)\ge\sum_ip_i\ln(p_i/q_i)=\KL(p\|q)$, with equality only if, for every $i$, the ratio $r_j/s_j$ equals $p_i/q_i$ for all $j$ with $r_jW(i|j)>0$.

(c) Fix $j$ with $r_j>0$. Every $i$ with $W(i|j)>0$ satisfies $p_i\ge r_jW(i|j)>0$, hence $i\in S$ and, by (a)--(b), $r_j/s_j=p_i/q^\star_i$. These ratios are distinct across $i\in S$, so exactly one $i$ has $W(i|j)>0$; then $W(i|j)=1$ and $e_j\in\operatorname{ran}P_i$. Therefore $\rho=\sum_{j:r_j>0}r_j|e_j\rangle\langle e_j|$ is a sum of operators each supported in a single eigenspace of $F$, so $[\rho,P_i]=0$ for all $i$, i.e.\ $[\rho,F]=0$.
\end{IEEEproof}

Corollary~\ref{cor:iff} closes the question of where the collective bracket $[1/\Dinf(k),1/\Kinf(k)]$ is degenerate: exactly at arms commuting with the reward observable. At every non-commuting arm the two floors differ, and whether the lower one is attainable is Open Problem~\ref{op:collective}.

\begin{openproblem}\label{op:collective}
Determine, per arm, the optimal $\ln T$ coefficient $c_k:=\inf_\pi\liminf_{T}\mathbb E_\pi[N_k(T)]/\ln T$, the infimum over consistent $\pi\in\Pi_{\mathrm{all}}$. Theorems~\ref{thm:match} and \ref{thm:univ} bracket it in $[\,1/\Dinf(k),\,1/\Kinf(k)\,]$. The certification core is a composite quantum Stein problem: testing the simple null $\rho_k^{\otimes n}$ against the composite i.i.d.\ alternative $\{\sigma^{\otimes n}:\sigma\in\Lambda_k\}$, with the exponent measured on the type-II side. General regularized formulas for this composite i.i.d. testing task are known \cite{BBH21,LamiSanov25}; in particular, Theorem~3 of \cite{LamiSanov25} covers the singleton null and closed confusing set here. The unresolved task is to evaluate the resulting exponent for the reward half-space, rather than to establish the existence of a composite testing theorem. Single-letter formulas fail in general \cite{MSW22}. The family $\{\sigma^{\otimes n}:\sigma\in\Lambda_k\}$ is not tensor-closed, so the generalized quantum Stein lemma for tensor-closed families \cite{Lami25,HY25} does not directly settle its value; see also \cite{Tangled24}. Whether the exponent equals $\Dinf(k)$ here --- and whether an optimal \emph{adaptive} regret protocol can realize that testing rate --- are separate open questions. Even a strict separation $c_k>1/\Dinf(k)$ at some instance would be significant: it would show the universal floor of Theorem~\ref{thm:univ} is not the right constant. Section~\ref{sec:n2} reports two-copy searches; whether a finite block size yields a collective advantage remains open.
\end{openproblem}

\section{Two-copy measurements: numerical evidence}\label{sec:n2}

We investigate whether jointly measuring pairs of copies can improve on the spectral measurement at the running qubit instance (defined in Section~\ref{sec:num}). The relevant optimizations are
\[
K^{(2)}(k):=\tfrac12\inf_{\sigma\in\Lambda_k}
D_M(\rho_k^{\otimes2}\|\sigma^{\otimes2})
\]
and
\[
V_2(k):=\sup_M\inf_{\sigma\in\Lambda_k}\tfrac12
\KL(P^M_{\rho_k^{\otimes2}}\|P^M_{\sigma^{\otimes2}}).
\]
The inf-sup envelope $K^{(2)}$ upper-bounds $V_2$; this inequality does not establish achievability or strict separation. Numerical optimization estimates $K^{(2)}=0.8718$, about $5\%$ above $\Kinf=0.8304$. The numerical minimizer lies in the plane spanned by the Bloch vector of $\rho_k$ and the reward axis ($b_y=0$); this observation is not a global optimization certificate.

\begin{remark}[Block hierarchy]\label{rem:blocks}
For block size $n$ let $V_n(k):=\sup_M\inf_{\sigma\in\Lambda_k}\frac1n\KL(P^M_{\rho_k^{\otimes n}}\|P^M_{\sigma^{\otimes n}})$, the supremum over $n$-copy POVMs, and $K^{(n)}(k):=\frac1n\inf_{\sigma\in\Lambda_k}D_M(\rho_k^{\otimes n}\|\sigma^{\otimes n})$. Then
\[
\Kinf(k)\ \le\ V_n(k)\ \le\ V_{2n}(k)\ \le\ K^{(2n)}(k)\ \le\ \Dinf(k).
\]
The first inequality uses the product-spectral measurement: its outcome law under $\sigma^{\otimes n}$ is $q(\sigma)^{\otimes n}$, so its per-copy divergence is $\KL(p^{(k)}\|q(\sigma))\ge\Kinf(k)$ for every $\sigma\in\Lambda_k$. The second uses the product $M\otimes M$ of an $n$-copy POVM with itself, whose per-copy divergence equals that of $M$ for every $\sigma$. The third is weak duality ($\sup\inf\le\inf\sup$, and $\KL(P^M_\cdot\|P^M_\cdot)\le D_M$). The last is $D_M\le D$ with additivity of $D$ on tensor powers. Nothing here is a regret statement for block-$n$ policies, and monotonicity $V_n\le V_{n+1}$ is not claimed.
\end{remark}

(i) \emph{An analytic lower bound.} Product-spectral measurement attains exactly $\Kinf$, hence $V_2\ge\Kinf$. For the projector reward used here, its distribution under $\sigma$ depends only on $\Tr\sigma F$. The two copies give independent Bernoulli outcomes, and minimizing their per-copy divergence over the cap gives $\kl(\mu_k,\mu^\star)=\Kinf$.

(ii) \emph{A transverse alternative defeats a promising candidate.} We numerically minimized $D_M(\rho_k^{\otimes2}\|\tau)$ over mixtures of tensor powers of confusing states in the symmetry plane ($b_y=0$ atoms) and took the eigenbasis of the resulting variational maximizer in \eqref{eq:bft} as a candidate measurement. Restricted to that two-parameter slice of alternatives, the inner search returns a value about $4\%$ above $\Kinf$. This is an estimated slice minimum, not a guaranteed full-set exponent. Over the full cap, the search finds an alternative near Bloch coordinates $(0,\pm0.40,0.90)$ with value $0.7355<\Kinf$. Such a feasible alternative upper-bounds the candidate's worst-case value and rules out an improvement by that candidate. Two implementations agree to $10^{-9}$ on the reported numerical value.

For $D(\rho\|\sigma)$ and $D_M(\rho\|\sigma)$, symmetry together with convexity permits restriction to the symmetry plane. For a fixed two-copy measurement, however, the probabilities depend on $\sigma$ through the nonlinear map $\sigma\mapsto\sigma^{\otimes2}$. That convexity argument does not apply. The transverse alternative illustrates the resulting risk of restricting the adversary to a slice.

(iii) \emph{Scope of the searches.} Beyond the product-spectral basis ($0.8304$) and the candidate of (ii), the tested measurements include the Bell basis ($0.345$), the symmetric/antisymmetric two-outcome measurement ($0$), and a seeded local max--min search over all two-copy projective bases $U\in U(4)$ ($16$ real parameters; Nelder--Mead followed by Powell; nine starts including the product basis), whose best worst-case value is $0.830366=\Kinf$ to $10^{-11}$, attained at the product-spectral basis. Grid or locally minimized adversarial values are upper bounds on the full-set infimum for each fixed basis. No tested basis improves on the analytic spectral value. The search does not cover all POVMs or certify the global maximum even over projective bases.

These observations support the conjecture $V_2=\Kinf$ at this instance, but do not prove it or a strict inf-sup/sup-inf gap. They also leave other fixed block sizes unresolved. Even a proof at two copies would not exclude an advantage at three or more copies.

\begin{remark}[Relation to the asymptotic problem]
For each fixed alternative $\sigma$, quantum Stein asymptotics give $n^{-1}D_M(\rho^{\otimes n}\|\sigma^{\otimes n})\to D(\rho\|\sigma)$ \cite{HiaiPetz91,OgawaNagaoka00}. This pointwise statement alone does not interchange the limit with optimization over alternatives or supply a single measurement for the whole confusing set. Composite testing and its reduction to an adaptive regret guarantee require separate arguments. The present experiments do not locate a possible collective advantage exclusively in the unbounded-block regime.
\end{remark}

\section{Numerical instantiation}\label{sec:num}

Take $F=|0\rangle\langle0|$ on a qubit; arm $k^\star$ with Bloch vector $(0,0,0.9)$, $\mu^\star=0.95$; arm $k$ with Bloch vector $(0.9,0,0)$, $\mu_k=0.5$, $\Delta_k=0.45$; $\Lambda_k=\{b_z\ge0.9,|b|\le1\}$. All values below are computed by seeded, two-route-verified scripts (nested minimization; independent variational and direct-POVM evaluations of $D_M$ agreeing to $\sim\!10^{-13}$).

\begin{table}[t]
\centering
\caption{Computed exponents at the running instance (nats).}
\label{tab:num}
\begin{tabular}{lc}
\toprule
$\kl(\mu_k,\mu^\star)=\Kinf(k)$ & $0.830366$\\
$C^*(k)$ (numeric; best deterministic POVM) & $0.830366$\\
$\Minf(k)$ (numeric, nested) & $0.830366$\\
$\Dinf(k)$ & $1.117330$\\
$K^{(2)}(k)$ (two-copy envelope, est.; Section~\ref{sec:n2}) & $0.871830$\\
Product-spectral (two copies; $V_2$ lower bound) & $0.830366$\\
pairwise $D_M(\rho_k\|\rho_{k^\star})$,\ $D(\rho_k\|\rho_{k^\star})$ & $1.059352$,\ $1.324998$\\
\bottomrule
\end{tabular}
\end{table}

\begin{figure}[t]
\centering
\includegraphics[width=\columnwidth]{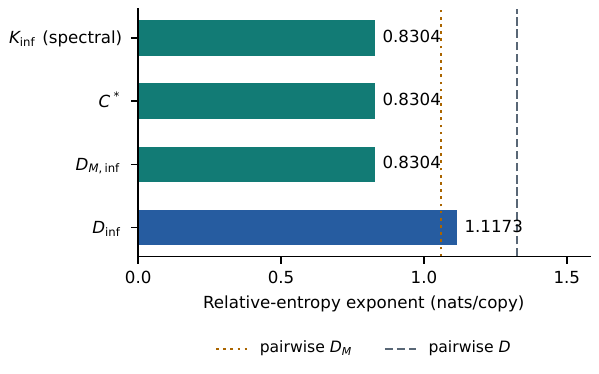}
\caption{The exponent ladder at the running instance. The three single-copy quantities coincide exactly (Theorem~\ref{thm:collapse}); the quantum-memory floor $\Dinf$ sits strictly above, with the pairwise divergences (dotted) shown for reference --- using them in place of the half-space infima would understate both floors.}
\label{fig:ladder}
\end{figure}

Three features of the collapse are worth recording. First, the nested minimizer of $D_M$ over $\Lambda_k$ landed at Bloch coordinates $(0.27971,0.90000)$; the analytic witness \eqref{eq:witness} gives $b_x^\star=2\rho_{01}\ln(\omega_1/\omega_0)/(\omega_1-\omega_0)=0.2797217$ with $\omega=\operatorname{diag}(10/19,10)$, and $D_M(\rho_k\|\sigma^\star)$ matches $\kl(\mu_k,\mu^\star)$ to $4\times10^{-12}$, with the numerical maximizer of \eqref{eq:bft} equal to $\operatorname{diag}(0.526315,10.000038)$ --- the theorem was, in effect, discovered by this coincidence and then proved. Second, the deterministic measurement game exhibits total mimicability of tilted measurements: the projective measurement at angle $\pi/4$ has game value $C^*_M=1.7\times10^{-13}$ against $\Lambda_k$, which is why the collapse $C^*=\Minf$ is not a priori obvious. Third, the regret slopes: the single-copy optimum $\Delta_k/\Kinf=0.5419$ is the constant that Bernoulli KL-UCB on the $F$-outcomes approaches (Fig.~\ref{fig:regret}), while the quantum-memory floor is $\Delta_k/\Dinf=0.4027$; quantum memory can save at most $25.7\%$ of the exploration at this instance --- if the floor is attainable at all (Open Problem~\ref{op:collective}; two-copy evidence in Section~\ref{sec:n2}).

\begin{figure}[t]
\centering
\includegraphics[width=\columnwidth]{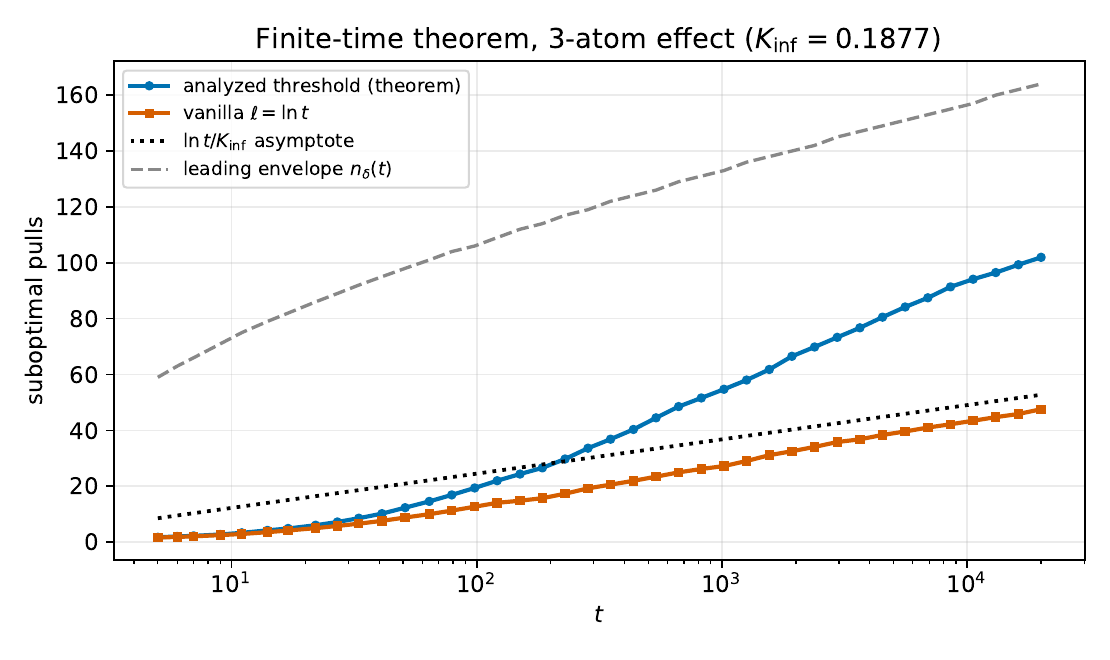}
\caption{Finite-time validation of Theorem~\ref{thm:finite} on a $3$-atom effect $f=(1,0.5,0)$, $q_{k^\star}=(0.55,0.30,0.15)$, $q_k=(0.30,0.30,0.40)$ ($\mu^\star=0.70$, $\mu_k=0.45$, $\Kinf(k)=0.187744$, exceeding $\kl(\mu_k,\mu^\star)=0.1346$ by $40\%$ --- the value of the full alphabet over the Bernoulli reduction). In a 100-run simulation (seed $7$, horizon $T=2\times10^4$), the empirical mean count stays below the leading envelope $n_\delta(t)$ (dashed) at every checkpoint under the analyzed threshold (solid); the lighter, unanalyzed threshold $\ell=\ln t$ climbs toward the predicted constant $1/\Kinf=5.33$ from below, illustrating the finite-time slack without establishing the limiting constant numerically.}
\label{fig:ftsim}
\end{figure}

Figure~\ref{fig:ftsim} instantiates Theorem~\ref{thm:finite} on a second instance with a genuinely $3$-symbol alphabet ($L=3$), confirming Lemma~\ref{lem:dual}'s dual form (matching a direct primal solve to $1.4\times10^{-16}$) and the constants $c_2=\ln2$, $c_3=\ln3$ (proved after Lemma~\ref{lem:mixture} and checked numerically through $n=400$).

\begin{figure}[t]
\centering
\includegraphics[width=\columnwidth]{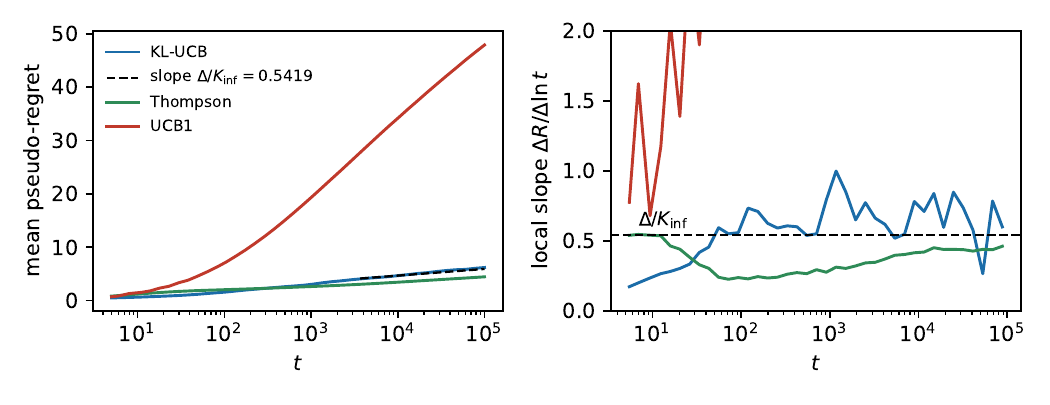}
\caption{Confirmatory regret simulation at the running instance ($5000$ independent runs, $T=10^5$; shading: $\pm2$ standard errors). Left: mean pseudo-regret for Bernoulli KL-UCB with the standard threshold $\ln t$ \cite{GarivierCappe11} on the spectral outcomes (the analyzed threshold of Section~\ref{sec:finite} adds $O(\ln\ln t)$ terms; both are compared in Fig.~\ref{fig:ftsim}), Thompson sampling \cite{KKM12}, and UCB1 \cite{ACF02}. Right: local slope $\Delta R/\Delta\ln t$ (UCB1 runs off the panel's scale). The tail OLS slope over $t\in[10^4,10^5]$ is $0.656$ (bootstrap $95\%$ CI $[0.639,0.675]$) for KL-UCB, decreasing toward the predicted $\Delta/\Kinf=0.5419$ (dashed); Thompson approaches the same constant from below ($0.437$); UCB1's measured tail slope is $5.97$ --- $9.1\times$ KL-UCB's, and $11\times$ the analytic optimum --- once Theorem~\ref{thm:collapse} fixes the measurement, the \emph{index} is what remains to get right.}
\label{fig:regret}
\end{figure}

\begin{figure}[t]
\centering
\includegraphics[width=\columnwidth]{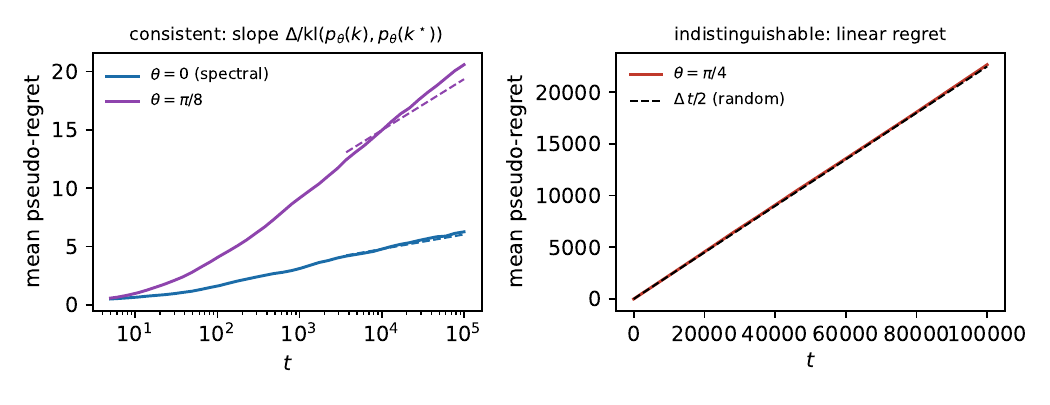}
\caption{Theorem~\ref{thm:collapse} operationally: fixing a projective measurement at angle $\theta$ from $F$'s eigenbasis induces a classical Bernoulli bandit with means $p_\theta(k^\star)=(1+0.9\cos\theta)/2$ and $p_\theta(k)=(1+0.9\sin\theta)/2$ ($1200$ runs each, same index). Left: at $\theta=0$ the KL-UCB constant is the quantum-optimal $\Delta/\Kinf=0.5419$; at $\theta=\pi/8$ it degrades to $\Delta/\kl(p_\theta(k),p_\theta(k^\star))=1.895$ (dashed asymptotes; measured tail slopes $0.63$ and $2.44$, each decreasing toward its constant). Right: at $\theta=\pi/4$ the induced means coincide ($(1+0.9/\sqrt2)/2$ for both arms), the arms are indistinguishable, and regret is linear, matching $\Delta t/2$ to $0.8\%$. Consistency over the full quantum instance class is stricter still for tilted measurements: the half-space constants $C^*_{M(\theta)}$ fall from $\Kinf=0.830$ at $\theta=0$ to $0.076$ at $\pi/8$ and to numerical zero at $\pi/4$ --- the mimicability noted above.}
\label{fig:meas}
\end{figure}

Figures~\ref{fig:regret} and~\ref{fig:meas} close the loop operationally. The first provides finite-horizon evidence consistent with Theorem~\ref{thm:match}; the second makes the collapse tangible by running the \emph{same} index on the outcomes of deliberately rotated measurements: the regret constant degrades continuously with the measurement angle and diverges at the mimicability point, while $\theta=0$ --- the spectral choice --- attains the optimum that Theorem~\ref{thm:lbsc} shows no consistent single-copy strategy can beat.

\section{Discussion}\label{sec:disc}

The message of Theorems~\ref{thm:collapse}--\ref{thm:match} is sharply practical: in single-copy sequential learning of unknown states with an observable-defined objective, measurement design is worth nothing asymptotically --- measure the objective, run a classical index, with the finite-time guarantee of Theorem~\ref{thm:finite}. Any improvement in the leading asymptotic regret coefficient requires collective processing, is capped by Corollary~\ref{cor:cap}, with the cap equal to one exactly when the arm commutes with $F$ (Corollary~\ref{cor:iff}); storing copies for later single-copy measurement does not help (Remark~\ref{rem:delay}); finite-time improvements are not excluded. A collective improvement has not been observed in the tested two-copy bases (Section~\ref{sec:n2}), and its ultimate attainability is a well-posed open question at the boundary of composite quantum hypothesis testing (Open Problem~\ref{op:collective}). Beyond linear reward functionals --- our motivation arises from learning-driven discovery loops over density-matrix models, where rewards are likelihoods --- a local linearization makes $F$ the gradient observable of the reward functional at the current model; formalizing this reduction is left to future work.

\section*{Acknowledgment}
The author thanks K.-M.~Chung (Academia Sinica) for helpful discussions.


\begin{thebibliography}{35}
\bibitem{BLT24} S.~Brahmachari, J.~Lumbreras, and M.~Tomamichel, ``Quantum contextual bandits and recommender systems for quantum data,'' \emph{Quantum Mach. Intell.}, vol.~6, art.~58, 2024.
\bibitem{SPJ26} J.~P.~Simpson, E.~Palias, and S.~T.~Jose, ``Optimal error exponents for composite sequential quantum hypothesis testing,'' arXiv:2605.04915v3, 2026.
\bibitem{MSW22} M.~Mosonyi, Zs.~Szil\'agyi, and M.~Weiner, ``On the error exponents of binary state discrimination with composite hypotheses,'' \emph{IEEE Trans. Inf. Theory}, vol.~68, no.~2, pp.~1032--1067, 2022.
\bibitem{LHT22} J.~Lumbreras, E.~Haapasalo, and M.~Tomamichel, ``Multi-armed quantum bandits: Exploration versus exploitation when learning properties of quantum states,'' \emph{Quantum}, vol.~6, p.~749, 2022.
\bibitem{LumPure24} J.~Lumbreras, M.~Terekhov, and M.~Tomamichel, ``Learning pure quantum states (almost) without regret,'' arXiv:2406.18370, 2024.
\bibitem{LumAnyDim26} J.~Lumbreras and M.~Tomamichel, ``Learning pure quantum states in any dimension (almost) without regret,'' arXiv:2605.09019, 2026.
\bibitem{LumThesis25} J.~Lumbreras, ``Bandits roaming Hilbert space,'' Ph.D.\ thesis, arXiv:2509.24569, 2025.
\bibitem{Wan23} Z.~Wan, Z.~Zhang, T.~Li, J.~Zhang, and X.~Sun, ``Quantum multi-armed bandits and stochastic linear bandits enjoy logarithmic regrets,'' in \emph{Proc. AAAI}, vol.~37, pp.~10087--10094, 2023.
\bibitem{LTT22} Y.~Li, V.~Y.~F. Tan, and M.~Tomamichel, ``Optimal adaptive strategies for sequential quantum hypothesis testing,'' \emph{Commun. Math. Phys.}, vol.~392, no.~3, pp.~993--1027, 2022.
\bibitem{BK96} A.~N. Burnetas and M.~N. Katehakis, ``Optimal adaptive policies for sequential allocation problems,'' \emph{Adv. Appl. Math.}, vol.~17, no.~2, pp.~122--142, 1996.
\bibitem{GarivierCappe11} A.~Garivier and O.~Capp\'e, ``The KL-UCB algorithm for bounded stochastic bandits and beyond,'' in \emph{Proc. COLT}, PMLR vol.~19, pp.~359--376, 2011.
\bibitem{MMS11} O.-A.~Maillard, R.~Munos, and G.~Stoltz, ``A finite-time analysis of multi-armed bandits problems with Kullback--Leibler divergences,'' in \emph{Proc. COLT}, PMLR vol.~19, pp.~497--514, 2011.
\bibitem{CGMMS13} O.~Capp\'e, A.~Garivier, O.-A.~Maillard, R.~Munos, and G.~Stoltz, ``Kullback--Leibler upper confidence bounds for optimal sequential allocation,'' \emph{Ann. Statist.}, vol.~41, no.~3, pp.~1516--1541, 2013.
\bibitem{ACF02} P.~Auer, N.~Cesa-Bianchi, and P.~Fischer, ``Finite-time analysis of the multiarmed bandit problem,'' \emph{Mach. Learn.}, vol.~47, no.~2--3, pp.~235--256, 2002.
\bibitem{KKM12} E.~Kaufmann, N.~Korda, and R.~Munos, ``Thompson sampling: An asymptotically optimal finite-time analysis,'' in \emph{Proc. Algorithmic Learning Theory (ALT)}, LNCS vol.~7568, pp.~199--213, Springer, 2012.
\bibitem{HondaTakemura15} J.~Honda and A.~Takemura, ``Non-asymptotic analysis of a new bandit algorithm for semi-bounded rewards,'' \emph{J. Mach. Learn. Res.}, vol.~16, pp.~3721--3756, 2015.
\bibitem{GMS19} A.~Garivier, P.~M\'enard, and G.~Stoltz, ``Explore first, exploit next: The true shape of regret in bandit problems,'' \emph{Math. Oper. Res.}, vol.~44, no.~2, pp.~377--399, 2019.
\bibitem{Umegaki62} H.~Umegaki, ``Conditional expectation in an operator algebra, IV (Entropy and information),'' \emph{K\=odai Math. Sem. Rep.}, vol.~14, pp.~59--85, 1962.
\bibitem{Lindblad75} G.~Lindblad, ``Completely positive maps and entropy inequalities,'' \emph{Commun. Math. Phys.}, vol.~40, pp.~147--151, 1975.
\bibitem{Uhlmann77} A.~Uhlmann, ``Relative entropy and the Wigner--Yanase--Dyson--Lieb concavity in an interpolation theory,'' \emph{Commun. Math. Phys.}, vol.~54, pp.~21--32, 1977.
\bibitem{HiaiPetz91} F.~Hiai and D.~Petz, ``The proper formula for relative entropy and its asymptotics in quantum probability,'' \emph{Commun. Math. Phys.}, vol.~143, pp.~99--114, 1991.
\bibitem{OgawaNagaoka00} T.~Ogawa and H.~Nagaoka, ``Strong converse and Stein's lemma in quantum hypothesis testing,'' \emph{IEEE Trans. Inf. Theory}, vol.~46, no.~7, pp.~2428--2433, 2000.
\bibitem{BFT17} M.~Berta, O.~Fawzi, and M.~Tomamichel, ``On variational expressions for quantum relative entropies,'' \emph{Lett. Math. Phys.}, vol.~107, no.~12, pp.~2239--2265, 2017.
\bibitem{JRSWW18} M.~Junge, R.~Renner, D.~Sutter, M.~M.~Wilde, and A.~Winter, ``Universal recovery maps and approximate sufficiency of quantum relative entropy,'' \emph{Ann. Henri Poincar\'e}, vol.~19, no.~10, pp.~2955--2978, 2018.
\bibitem{Petz86} D.~Petz, ``Sufficient subalgebras and the relative entropy of states of a von Neumann algebra,'' \emph{Commun. Math. Phys.}, vol.~105, pp.~123--131, 1986.
\bibitem{BBH21} M.~Berta, F.~G. S. L. Brand\~ao, and C.~Hirche, ``On composite quantum hypothesis testing,'' \emph{Commun. Math. Phys.}, vol.~385, pp.~55--77, 2021.
\bibitem{LamiSanov25} L.~Lami, ``Generalised quantum Sanov theorem revisited,'' arXiv:2510.06340v1, 2025.
\bibitem{Lami25} L.~Lami, ``A solution of the generalised quantum Stein's lemma,'' \emph{IEEE Trans. Inf. Theory}, vol.~71, no.~6, pp.~4454--4484, 2025.
\bibitem{HY25} M.~Hayashi and H.~Yamasaki, ``The generalized quantum Stein's lemma and the second law of quantum resource theories,'' \emph{Nat. Phys.}, vol.~21, pp.~1988--1993, 2025.
\bibitem{BHKW20} M.~M.~Wilde, M.~Berta, C.~Hirche, and E.~Kaur, ``Amortized channel divergence for asymptotic quantum channel discrimination,'' \emph{Lett. Math. Phys.}, vol.~110, no.~8, pp.~2277--2336, 2020.
\bibitem{Tangled24} M.~Berta, F.~G. S.~L.~Brand\~ao, G.~Gour, L.~Lami, M.~B. Plenio, B.~Regula, and M.~Tomamichel, ``The tangled state of quantum hypothesis testing,'' \emph{Nat. Phys.}, vol.~20, pp.~172--175, 2024.
\bibitem{KrichevskyTrofimov81} R.~E. Krichevsky and V.~K. Trofimov, ``The performance of universal encoding,'' \emph{IEEE Trans. Inf. Theory}, vol.~27, no.~2, pp.~199--207, 1981.
\bibitem{XieBarron00} Q.~Xie and A.~R. Barron, ``Asymptotic minimax regret for data compression, gambling, and prediction,'' \emph{IEEE Trans. Inf. Theory}, vol.~46, no.~2, pp.~431--445, 2000.
\bibitem{Ville39} J.~Ville, \emph{\'Etude critique de la notion de collectif}. Gauthier-Villars, Paris, 1939.
\bibitem{WOSVW03} T.~Weissman, E.~Ordentlich, G.~Seroussi, S.~Verd\'u, and M.~J. Weinberger, ``Inequalities for the $L_1$ deviation of the empirical distribution,'' Hewlett-Packard Labs, Tech.\ Rep.\ HPL-2003-97, 2003.
\end{thebibliography}
\end{document}